\documentclass{article}%
\usepackage{amsmath}
\usepackage{amsfonts}
\usepackage{amssymb}%
\usepackage{graphicx,multicol}
\usepackage{hyperref}

\newtheorem{theorem}{Theorem}
\newtheorem{acknowledgement}[theorem]{Acknowledgement}
\newtheorem{definition}[theorem]{Definition}
\newtheorem{example}[theorem]{Example}
\newtheorem{lemma}[theorem]{Lemma}
\newtheorem{proposition}[theorem]{Proposition}
\newtheorem{corolary}[theorem]{Corollary}

\newenvironment{proof}[1][Proof]{\noindent\textbf{#1.} }{\ \rule{0.5em}{0.5em}}
\begin{document}

\title{Composition operators on weighted spaces of holomorphic functions on
$JB^{\ast}-$triples}
\author{Michael Mackey
\and Pablo Sevilla-Peris
\and Jos\'e A. Vallejo}
\date{}
\maketitle

\begin{abstract}
We characterise continuity of composition operators on weighted spaces of
holomorphic functions $H_{v}(B_{X})$, where $B_{X}$ is the open unit ball of a
Banach space which is homogeneous, that is, a $JB^{\ast}-$triple.

\end{abstract}

\section{Introduction}

In this note, we prove a  result concerning composition
operators on $JB^{\ast}-$triples. These triples are Banach spaces
which carry a certain algebraic structure. They form quite a large
class, including Hilbert
spaces and $C^{\ast}-$algebras (see example \ref{example1} below), and are
interesting from both the mathematical and the physical point of view. On the
mathematical side, they play a r\^{o}le similar to that of semisimple Lie
algebras in the study of symmetric finite dimensional manifolds, but in the
context of infinite dimensional spaces (see \cite{Up85} and references
therein). Also, $JB^{\ast}-$triples are intimately related to Jordan algebras,
which are long known to appear in quantum mechanics (see
\cite{Duf38,Kem39,Jac49}, or \cite{Up92} for a recent account).
$JB^{\ast}-$triples have been found to be useful in solving
Yang-Baxter equations (\cite{Ok93}), constructing Lie superalgebras (see
\cite{kAoK04} and \cite{Sal05}) and in the study of multifield integrable systems (see
\cite{AdSviYam99} or \cite{Svi92} and references therein).

With respect to composition operators, let us recall that on a classical level
the coherent states of a physical system are described by holomorphic
functions on the classical phase space (see \cite{Ber74}). When passing to the
quantum framework, one deals with the general concept of state over
$\mathcal{B}(\mathcal{H})$ (the algebra of bounded linear operators on a
Hilbert space $\mathcal{H}$), which is a normalized positive linear functional
on $\mathcal{B}(\mathcal{H})$ (see \cite{Ara99}). In these contexts,
composition operators can be seen as \textquotedblleft
dictionaries\textquotedblright\ translating these states from one reference
frame to another when we have a holomorphic transformation between the
underlying spaces $\phi:X\rightarrow Y$ (in this case, the composition
operator associated to $\phi$, $C_{\phi}$, is a map $C_{\phi}:H(Y)\rightarrow
H(X)$, where $H(X)$ is the space of holomorphic mappings from $X$ to
$\mathbb{C}$).

Both situations (the classical and the quantum ones), are generalized in the
study of weighted spaces of holomorphic functions on the unit ball $B$ of a Banach space $X$,
denoted $H_{v}(B)$. These spaces have been widely studied in recent years, and
are quite well understood. The first case considered was that of $B$
being the unit disc or a domain in $\mathbb{C}$ or $\mathbb{C}^{n}$.  Special
interest has been given to the study of composition operators between these
spaces; we refer to
\cite{BiBoGa93,BoDoLiTa98,CoHe00} and particularly to the recent surveys \cite{Bi02,Bo03} and the references therein
for information about the
subject.\ Some study has also been devoted to the situation when $B_{X}$ is
the open unit ball of a Banach space $X$ (see e.g.
\cite{ArGaLi97,GaMaRu00,GaMaSe04}).  Some of the results in \cite{BoDoLiTa98} were generalised in
\cite{GaMaSe04} to the Banach space setting.  One result given in
\cite{GaMaSe04} characterizes continuity of composition operators when
$B$ is the open unit ball of a Hilbert space.  The proof relies on the
fact that there exist enough automorphisms of $B$.  In this note, we
show that this requirement is also fulfilled if we consider unit balls
of $JB^{\ast}-$triples.\newline

\section{Preliminary results}

We begin by fixing notation and some results; for details see \cite{GaMaSe04}.
Let $X$ be a Banach space and $B_{X}$ its open unit ball.\ By a weight we
mean any continuous bounded mapping $v:B_{X}\rightarrow]0,\infty\lbrack$.\ We
denote by $H(B_{X})$ the space of holomorphic functions $f:B_{X}%
\longrightarrow\mathbb{C}$.\ A set $A\subset B_{X}$ is said to be $B_{X}%
$-bounded if $d(A,X\setminus B_{X})>0$.\ The subspace of $H(B_{X})$
consisting of those
functions which are bounded on the $B_{X}$-bounded sets is denoted by
$H_{b}(B_{X})$.\ Following \cite{BoDoLiTa98} and \cite{GaMaRu00} we consider
\[
H_{v}(B_{X})=\{f\in H(B_{X}):\Vert f\Vert_{v}=\sup_{x\in B_{X}}%
\ v(x)|f(x)|<\infty\},
\]
where $v$ is a weight.\ With the norm $\Vert\ \Vert_{v}$, the space
$H_{v}(B_{X})$ is a Banach space.\newline

Given a weight $v$, we consider the associated weight $\tilde{v}%
(x)=1/\sup_{\Vert f\Vert_{v}\leq1}|f(x)|$ (see \cite{BiBoGa93,BoDoLiTa98,GaMaSe04}).\ We
say that a weight $v$ is norm-radial if $v(x)=v(y)$ for every
$x$, $y$ such that $\Vert x\Vert=\Vert y\Vert$.\ If $v$ is norm-radial and
non-increasing (with respect to the norm) then $\tilde{v}$ is also norm-radial
and non-increasing.

A weight $v$ satisfies Condition I if $\inf_{x\in rB_{X}}v(x)>0$ for every
$0<r<1$ (\cite{GaMaRu00}).\ If $v$ satisfies Condition I, then $H_{v}%
(B_{X})\subseteq H_{b}(B_{X})$ (\cite[Proposition 2]{GaMaRu00}).\newline

\begin{definition}
Let $X$ and $Y$ be Banach spaces and
$\phi:B_{X}\rightarrow B_{Y}$ a holomorphic mapping.\ The composition
operator associated to $\phi$ is defined by
\[
C_{\phi}:H(B_{Y})\longrightarrow H(B_{X})\;\;\;,\;\;f\rightsquigarrow C_{\phi
}(f)=f\circ\phi.
\]

\end{definition}

$C_{\phi}$ is clearly linear.\ Denoting by $\tau_{0}$ the compact-open
topology, $C_{\phi}$ is also $(\tau_{0},\tau_{0})$-continuous.\ Given any two
weights $v_{X},v_{Y}$ defined on $B_{X},B_{Y}$ respectively, we consider the
restriction $C_{\phi}:H_{v_{Y}}(B_{Y})\rightarrow H_{v_{X}}(B_{X})$ whenever
this is well defined.\ It is known that if $C_{\phi}$ is well defined, then it
is continuous (see \cite{GaMaSe04}). The following result was proved
in \cite{GaMaSe04} (see also \cite[Proposition 2.1]{BoDoLiTa98}).

\begin{proposition}
\label{contgen} Let $v_{X}$, $v_{Y}$ be two weights satisfying Condition I and
$\phi: B_{X} \longrightarrow B_{Y}$ holomorphic.\ Then the following are
equivalent,\newline$(i)$ \ \ $C_{\phi} : H_{v_{Y}}(B_{Y}) \longrightarrow
H_{v_{X}}(B_{X})$ is well defined and continuous.\newline$(ii)$
\ $\displaystyle\sup_{x\in B_{X}}\displaystyle\frac{v_{X}(x)}{\tilde{v}_{Y}
(\phi(x))}< \infty$.\newline$(iii)$ $\displaystyle\sup_{x\in B_{X}}
\displaystyle\frac{\tilde{v}_{X}(x)}{\tilde{v}_{Y}(\phi(x))}< \infty$%
.\newline$(iv)$ \ $\displaystyle\sup_{\|\phi(x)\| > r_{0}} \frac{v_{X}%
(x)}{\tilde{v}_{Y} (\phi(x))} < \infty$ for some $0<r_{0}<1$.
\end{proposition}

\section{$JB^{\ast}-$triples}

We intend to study composition operators on a $JB^{\ast}-$triple $X$.
In this case, $B_{X}$ is a bounded symmetric
domain.\ Given a domain $D$ in a Banach space, a symmetry at $a\in D$ is a
biholomorphic map $s_{a}:D\rightarrow D$ such that $s_{a}^{2}=id$ and
$s_{a}(a)=a$ is an isolated fixed point.  A bounded symmetric domain is a
bounded domain (or a domain biholomorphically equivalent to a bounded domain)
which has a symmetry at every point.

\begin{definition}
\label{JB*} A $JB^{\ast}-$triple is a Banach space $Z$ with a triple product
$\{\ ,\ ,\ \}:Z^{3}\longrightarrow Z$ that is linear and symmetric in the
first and third variables (symmetric in the sense that $\{x,y,z\}=\{z,y,x\}$
for all $x,z$) and antilinear in the second variable and which
satisfies,\newline$(i)$ the mapping $x\Box x$, given by $x\Box x(z)=\{x,x,z\}$
is Hermitian, $\sigma(x\Box x)\geq0$ and $\Vert x\Box x\Vert=\Vert x\Vert^{2}%
$,\newline$(ii)$ for every $a,b,x,y,z\in X$, the Jordan triple identity
\[
\{a,b,\{x,y,z\}\}=\{\{a,b,x\},y,z\}-\{x,\{b,a,y\},z\}+\{x,y,\{a,b,z\}\}
\]
holds.
\end{definition}

For  $x,y\in Z$, we define three mappings $x\Box y$ (linear), $Q_{x}$
(antilinear) and $B(x,y)$ (linear) by
\begin{gather*}
x\Box y(z)=\{x,y,z\},\\\;Q_{x}(z)=\{x,z,x\},\\ B(x,y)=id-2x\Box y+Q_{x}Q_{y}.
\end{gather*}
We also consider the operator $B_{x}=B(x,x)^{1/2}$ (the square root taken in
the sense of functional calculus, i.e. $B_{x}\circ B_{x}=B(x,x)$).\ It is
known that (\cite{Ka94})
\begin{equation}
\Vert B_{x}^{-1}\Vert=\frac{1}{1-\Vert x\Vert^{2}}.\label{normaBx}%
\end{equation}

For background on $JB^{\ast}-$triples, see \cite{Ha81,Me03}.\newline

It is a well known fact that the open unit ball of a Banach space is symmetric
if and only if the space is a $JB^{\ast}-$triple \cite{Ka83}.\ Also, a bounded
domain $D$ is symmetric if and only if it has a transitive group of
biholomorphic mappings $\{g_{a}\}_{a\in D}$ and a symmetry at some point
$p$.\ In this case the bounded symmetric domain is biholomorphically
equivalent to the unit ball of a $JB^{\ast}-$triple and all biholomorphic
mappings on the unit ball can be explicitly described.\ They are of the form
$Kg_{a}$ where $K$ is a surjective linear isometry and $g_{a}$ are M\"{o}bius
type mappings that satisfy $g_{a}(0)=a$ and $g_{a}^{-1}=g_{-a}$ (\cite{Ka94}%
).\ These mappings can be defined from the triple product by
\begin{align*}
g_{a}(x) &  =a+(B(a,a)^{1/2}\circ B(x,a)^{-1})(x-Q_{x}(a))\\
&  =a+B_{a}(\sum_{n=0}^{\infty}(-x\Box a)^{n}a)
\end{align*}
If $s_{0}$ denotes the symmetry at $0$ (i.e. $x\mapsto-x$), the symmetry at
any other point of the unit ball $a$ is given by $g_{a}\circ s_{0}\circ
g_{-a}$.\newline

\begin{example}
\label{example1}Examples of $JB^{\ast}-$triples are
Hilbert spaces and $C^{\ast}-$algebras.\ On a Hilbert space the triple product
is given by $\{x,y,z\}=1/2((x|y)z+(z|y)x)$.  The M\"{o}bius mappings for Hilbert spaces were defined by Renaud in
\cite{Re73}.\ If $Z$ is a $C^{\ast}-$algebra, the triple product is given by
$\{x,y,z\}=1/2(xy^{\ast}z+zy^{\ast}x)$.\ Another example of $JB^{\ast}%
-$triples that includes the two previous ones are $J^{\ast}$-algebras, that is
closed subspaces of $\mathcal{L}(H,K)$ ($H$ and $K$ Hilbert spaces) which are
closed under $A\mapsto AA^{\ast}A$ (cf.~\cite{Ha81}).\newline
\end{example}

As already mentioned, the symmetries of a bounded symmetric domain can be
defined using a set of M\"{o}bius-like mappings.\ Let us show that these
vector M\"{o}bius mappings behave in the same way as the scalar ones when we
take the supremum on a sphere (a circle in the scalar case).

\begin{lemma}
\label{supga} Let $B$ be a bounded symmetric domain (i.e., the open unit ball
of a $JB^{\ast}-$triple $Z$) and $\{g_{a}\}_{a\in B}$ the transitive group of
biholomorphic mappings that define the symmetries.\ Then, for each $0<r<1$
\[
\sup_{\Vert x\Vert=r}\Vert g_{a}(x)\Vert=\frac{\Vert a\Vert+r}{1+r\Vert
a\Vert}%
\]
and this supremum is attained at some point.
\end{lemma}

\begin{proof}
First, for any bounded symmetric domain we show that $\Vert g_{a}(x)\Vert
\leq\frac{\Vert a\Vert+\Vert x\Vert}{1+\Vert a\Vert\cdot\Vert x\Vert}$.\ It is
well known (\cite{Me03}) that
\[
\frac{1}{1-\Vert g_{a}(x)\Vert^{2}}=\Vert B_{a}^{-1}\circ B(a,x)\circ
B_{x}^{-1}\Vert.
\]
In particular, using \eqref{normaBx} we get
\[
\frac{1}{1-\Vert g_{a}(x)\Vert^{2}}\leq\frac{1}{1-\Vert a\Vert^{2}}\ (1+\Vert
a\Vert\cdot\Vert x\Vert)^{2}\ \frac{1}{1-\Vert x\Vert^{2}}.
\]
Hence
\[
\Vert g_{a}(x)\Vert\leq\frac{\Vert a\Vert+\Vert x\Vert}{1+\Vert a\Vert
\cdot\Vert x\Vert}.
\]

Next we show that the bound is attained, in the sense that there
exists $x\in B$, $\Vert x\Vert=r$ with $\Vert g_{a}(x)\Vert=\frac{\Vert
a\Vert+r}{1+r\ \Vert a\Vert}$.\ Clearly we may assume $a\neq0$.\ Let us
consider $Z_{a}$ the $JB^{\ast}-$subtriple of $Z$ generated by $a$, that is,
the smallest (closed) $JB^{\ast}-$subtriple of $Z$ that contains $a$.\ It is
obviously enough to find $x\in Z_{a}$ attaining the bound.\ A result of
Kaup (\cite[Proposition 5.3]{Ka83}) shows that for any $JB^{\ast}$-triple and
$a\in Z$, $Z_{a}$ is isometrically (triple) isomorphic to $C_{0}(\Omega)$,
where $\Omega\subseteq \mathbb{R}^{+}$ satisfies $\Omega\cup\{0\}$ is compact.\ The
M\"{o}bius maps on the unit ball of $Z_{a}$, once composed with this
isomorphism, give $g_{a}(z)=\frac{a+z}{1+\bar{a}\ z}$, where $a$ and $z$ are
in the open unit ball of $C_{0}(\Omega)$. For $z=\frac{r}{\Vert a\Vert}a$, we
have $z\in C_{0}(\Omega)$ and $\Vert z\Vert=r$.\ Hence
\[
g_{a}(z)=\frac{\left(  1+\frac{r}{\Vert a\Vert}\right)  \ a}{1+|a|^{2}%
\ \frac{r}{\Vert a\Vert}}=\frac{r+\Vert a\Vert}{\Vert a\Vert+r\ |a|^{2}}\ a.
\]
Now, $\Vert g_{a}(z)\Vert=(r+\Vert a\Vert)\ \left\Vert \frac{a}{\Vert
a\Vert+r\ |a|^{2}}\right\Vert =(r+\Vert a\Vert)\ \sup_{\omega\in\Omega}%
\frac{|a|}{\Vert a\Vert+r\ |a|^{2}}(\omega)$.\ But since $|a|\leq\Vert
a\Vert\leq1$ and $r<1$, it turns out that $\frac{|a|}{\Vert a\Vert+r\ |a|^{2}%
}$ is an increasing function of $|a|$, that is $\left\Vert \frac{a}{\Vert
a\Vert+r\ |a|^{2}}\right\Vert =\frac{1}{1+r\ \Vert a\Vert}$.\ This gives
\[
\Vert g_{a}(z)\Vert=\frac{\Vert a\Vert+\Vert z\Vert}{1+\Vert a\Vert\cdot\Vert
z\Vert}
\]
which is what was required.
\end{proof}

\section{A result for composition operators}

The following result is a very well known version of the Schwarz lemma for
Banach spaces (cf. \cite{DinSch}).\newline

\begin{lemma}
Let $X$ and $Y$ be Banach spaces and  $f:B_{X}\longrightarrow B_{Y}$
holomorphic with $f(0)=0$.
Then, for all $x\in B_{X}$,
\[
\Vert f(x)\Vert_{Y}\leq\Vert x\Vert_{X}.
\]

\end{lemma}

We can now prove a generalization of \cite[Theorem 2.3]{BoDoLiTa98} and
\cite[Theorem 4.1]{GaMaSe04}.\ The statement is slightly different from
the previous cases but the proof is basically the same, up to technical
changes.\ We include a proof for the sake of completeness.

\begin{theorem} \label{contJB*triple}
Let $X$ be any Banach space and $Z$ a $JB^{\ast}-$triple.\ Let $v_{Z}$ be a
norm-radial and non-increasing weight on $Z$ and $v_{X}$ be a weight on $X$ for which
there exists $K>0$ such that \\
\centerline{if
$z\in Z$ and $x\in X$ with $\Vert z\Vert\leq\Vert x\Vert$, then
\(
v_{Z}(z)\geq Kv_{X}(x).
\)}
Then every composition operator $C_{\phi}:H_{v_{Z}} (B_{Z})\longrightarrow H_{v_{X}}(B_{X})$
is continuous for every holomorphic map $\phi:B_{X}\rightarrow B_{Z}$ if and only if the function
$l(r) := \tilde{v}_{Z}(z)$ for $\Vert z \Vert  = 1-r$, $0 < r < 1$ satisfies $l(s) \leq  M l(s/2)$
for $s$ close enough to 0.
\end{theorem}

\begin{proof}
First, if $\phi(0) = 0$ then by the general version of the Schwarz Lemma
we have $\| \phi(x) \|_{Z} \leq\|x\|_{X}$ and $C_{\phi}$ is
continuous.\ For each $a \in B_{Z}$ we have $g_{a} : B_{Z} \rightarrow
B_{Z}$.\ If every
$C_{g_{a}}$ is continuous then all $C_{\phi}$ are continuous.\ Indeed, given
$\phi$, let $a = \phi(0)$ and define $\psi= g_{-a} \circ\phi$.\ Then
$\psi(0)=0$ and $C_{\phi} = C_{\psi} \circ C_{g_{a}} $ is
continuous.\ Therefore it is enough to prove that $C_{g_{a}}: H_{v_{Z}}(B_{Z})
\rightarrow H_{v_{Z}}(B_{Z})$ is continuous for all $a \in B_{Z}$ if and only
if, for all $0< s < s_{0}$,
\begin{equation}\label{contH}
l(s) \leq  M l(s/2)
\end{equation}

Assume that all $C_{g_{a}}$ are continuous.\ By Proposition \ref{contgen}, for each $a \in
B_{Z}$ we can find $M_{a} >0$ such that $\tilde{v}_{Z}(z) \leq M_{a} \tilde
{v}_{Z} (g_{a} (z))$ for all $z \in B_{Z}$.\ We also know that $\sup_{\|z\|=r}
\| g_{a}(z)\| = \frac{\| a \| + r}{1 + r \|a\|}$.\ Since $v_{Z}$ is norm-radial and non-increasing so also
is $\tilde{v}_{Z}$.\ Hence the previous can be rewritten as
\[
l(1-r) \leq M_{a} l\left(1-\frac{\| a \| + r}{1 + r \|a\|}\right)
= M_{a} l\left( \frac{(1-r)(1-\|a\|)}{1+r\|a\|}\right).
\]
\noindent Now, for $1/2 < r <1$ we have
\begin{equation}
\label{desigl}
l \left(  (1-r) \ \frac{1- \|a\|}{1+ \|a\|}\right)  \leq l \left(
1-\frac{\| a \| + r}{1 + r \|a\|} \right)
\leq l \left(  (1-r) \ \frac{1-\|a\|}{1+ \|a\|/2} \right)  .
\end{equation}
Let us fix $a$ with $\|a\|=2/5$ and use the second inequality in \eqref{desigl} to
get $l(1-r) \leq M_{a} l\left( \frac{(1-r)(1-\|a\|)}{1+r\|a\|}\right) \leq
M_{a} l(\frac{1-r}{2})$ for $1/2 < r <1$.\ This shows that \eqref{contH} holds.

Let us assume now that \eqref{contH} holds.\ Given any $c>0$ we can choose $n \in\mathbb{N}$
with $c<2^{n}$.\ If $s<s_{0}$, then $l(s) \leq K^{n}\ l(s/c)$.\ Given any
$a \in B_{Z}$, let us take $c=\frac{1+\|a\|}{1-\|a\|}$ and use the first inequality in \eqref{desigl} to
get that there exists $K_{a} >0$ such that holds.
\[
l(s) \leq K_{a} l(s/c) \leq K_{a} l \left(  1-\frac{\| a \| + (1-s)}{1 + (1-s) \|a\|} \right)
\]
for $s<s_{0} \leq 1/2$.\ Now, for $s_{0} \leq t \leq 1$, since $l$ is strictly positive,
the mapping $s \rightsquigarrow (l(s))(l(1-\frac{\|a\|(1-s)}{1+(1-s)\|a\|}))^{-1}$
is well defined and continuous; hence it has a maximum.\  Thus for any fixed $a \in B_{Z}$ we can find a
constant $M_{a} > 0$ such that for $0<r<1$ and $\|z\| = r$,
\[
\tilde{v}_{Z}(z) \leq M_{a} \ l \left(  1-\frac{\| a \| + r}{1 + r \|a\|}
\right)  \leq M_{a} \ \tilde{v}_{Z}(g_{a}(z)).
\]
Applying Proposition \ref{contgen}, $C_{g_{a}}$ is continuous.
\end{proof}

Several equivalent conditions on a weight $v$ so that $l$
satisfies \eqref{contH} are given in \cite[Lemma 1]{DoLi02} for the one-dimensional
case.\ Most of the proofs can be trivially adapted to the infinite dimensional case.\\

\noindent By taking  $X=Z$ and $v_{X}=v_{Z}$ in Theorem \ref{contJB*triple} we get
\begin{corolary}
Let $v$ be a norm-radial and non-increasing weight on a $JB^{\ast}-$triple $Z$.\ Every composition operator
$C_{\phi}$ on the weighted Banach space $H_{v} (B_{Z})$ is continuous for every self map $\phi$ on $B_{Z}$
if and only if the function $l(r) := \tilde{v}(z)$ for $\Vert z \Vert  = 1-r$, $0 < r < 1$ satisfies $l(s) \leq  M l(s/2)$
for $s$ close enough to 0.
\end{corolary}

\begin{acknowledgement}
The authors wish to thank S. Dineen, D. Garc\'{\i}a and M. Maestre for all
their helpful comments.\ We also wish to thank the referees for their comments,
specially those regarding the statement and proof of Theorem \ref{contJB*triple}
and for drawing our attention to \cite{DoLi02}
\end{acknowledgement}

\begin{multicols}{3}
\raggedright
\small
\vbox{
\noindent Michael Mackey \\
Dept. of Mathematics\\
University College Dublin\\
Belfield, Dublin 4. Ireland\\$\mathtt{michael.mackey@ucd.ie}$ \\
}

\vbox{
\noindent Pablo Sevilla-Peris \\
Departamento de Matem\'{a}tica Aplicada\\
ETSMRE,\\ Universidad Polit\'{e}cnica de Valencia\\Av. Blasco Ib\'{a}\~{n}ez 21,\\ 46010 Valencia. Spain\\
$\mathtt{pablo.sevilla@uv.es}$\\
}

\vbox{
\noindent Jos\'e A. Vallejo \\
Dep. Matem\`atica Aplicada IV \\
Universitat Polit\`ecnica de Catalunya \\
Avda. del Canal Olímpic, s/n \\
08860 Castelldefels Spain \\
$\mathtt{jvallejo@ma4.upc.edu}$
}
\end{multicols}


\begin{thebibliography}{99}
\bibitem {AdSviYam99}Adler,V. E.,  Svinolupov, S. I. and  Yamilov, R. I.:
Multi-component Volterra and Toda type integrable equations. \emph{Phys.
Lett. A} \textbf{254} (1999), 24-36.

\bibitem {Ara99} Araki, H.: \emph{Mathematical theory of quantum fields}. Oxford
University Press, 1999.

\bibitem {ArGaLi97}Aron, R.~M. ,~Galindo, P. and~Lindstr\"{o}m, M.: Compact
homomorphisms between algebras of analytic functions. \emph{Studia Math.}
\textbf{123}, 3 (1997), 235--247.

\bibitem {Ber74}Berezin, F. : Quantization. \emph{Math. USSR Izv.} \textbf{8}, 5
(1974) 1109-1165.

\bibitem {Bi02}Bierstedt, K.~D.: A survey of some results and open
problems in weighted inductive limits and projective description for spaces of
holomorphic functions. \emph{Bull. Soc. Roy. Sci. Li\`{e}ge} \textbf{70} (2002), 167--182.

\bibitem {BiBoGa93}Bierstedt, K.~D. ,~Bonet, J. and ~Galbis, A.: Weighted
spaces of holomorphic functions on balanced domains. \emph{Michigan Math.}
\textbf{40} (1993), 271--297.


\bibitem {Bo03}~Bonet, J.: Weighted spaces of holomorphic functions and
operators between them. In \emph{Proceedings of the seminar of Mathematical
Analysis (U. of M\'{a}laga and U. de Sevilla)}. Secretariado de Publicaciones,
Universidad de Sevilla, 2003. 117-138.


\bibitem {BoDoLiTa98}Bonet, J.~, Doma\'{n}ski, P.~, Lindstr\"{o}m, M.~ and
Taskinen, J.~: Composition operators between weighted Banach spaces of
analytic functions. \emph{J. Austral. Math. Soc. (Series A)} \textbf{64} (1998), 101--118.

\bibitem {CoHe00} Contreras, M.~D. and Hern\'{a}ndez-D{\'{\i}}az, A.~:
Weighted composition operators in weighted Banach spaces of analytic
functions. \emph{J. Austral. Math. Soc. (Series A)} \textbf{69} (2000), 41--60.

\bibitem {DinSch}Dineen, S.~: \emph{The Schwarz lemma}. Oxford University Press, 1989.

\bibitem{DoLi02}Doma\'nski, P. and Lindstr\"{o}m, M.~: Sets of interpolation and
sampling for weighted Banach spaces of holomorphic functions,
\emph{Ann. Pol. Math.} \textbf{79} (2002) 233--264.

\bibitem {Duf38}Duffin, R. J.: On the characteristic matrices of
covariant systems. \emph{Phys. Rev.} \textbf{54} (1938), 1114.

\bibitem {GaMaRu00}Garc\'{\i}a, D.~, Maestre, M.~ and Rueda, P.: Weighted
spaces of holomorphic functions on Banach spaces. \emph{Studia Math.} \textbf{138},
1 (2000), 1--24.

\bibitem {GaMaSe04}Garc\'{\i}a, D.~, Maestre, M.~ and Sevilla-Peris, P.~:
Composition operators between weighted spaces of holomorphic functions
on Banach spaces. \emph{Ann. Acad. Sci. Fenn. Math.} \textbf{29} (2004), 81-98.


\bibitem {Ha81} Harris, L.~A.: A generalization of $C^{\ast}%
-$\emph{algebras}. \emph{Proc. London Math. Soc.} \textbf{42} (1981), 331--360.

\bibitem {Jac49}Jacobson, N. : Lie and Jordan triple systems. \emph{Am. J. of
Math.} \textbf{71}, 1 (1949), 149-170.

\bibitem {kAoK04}Kamiya, N.  and Okubo, S. : A construction of simple Lie
superalgebras of certain types from triple systems. \emph{Bull. Aus. Math. Soc.}
\textbf{69}, 1 (2004), 113-123.

\bibitem {Ka83}Kaup, W.~: Riemann mapping theorem for bounded symmetric
domains in complex Banach spaces. \emph{Math. Zeit.} \textbf{183} (1983), 503--529.

\bibitem {Ka94}Kaup, W.~: Hermitian Jordan triple systems and the
automorphisms of bounded symmetric domains. In \emph{Third international
conference on non associative algebra and its applications} (Oviedo, July
12-17, 1993). Kluwer Math. Appl. \textbf{303}, Dordrecht (1994) 204--214.

\bibitem {Kem39}Kemmer, N.: Particle aspects of meson theory. \emph{Proc. of
the Royal Soc.} \textbf{173} (1939), 91-116.

\bibitem {Me03}Mellon, P.~: Holomorphic invariance on bounded symmetric
domains. \emph{J. Reine Angew. Math.} \textbf{523} (2000), 199--223.


\bibitem {Ok93}Okubo, S.: Triple products and Yang-Baxter equation I:
octonionic and quaternionic triple systems. \emph{J. Math. Phys.} \textbf{34}, 7
(1993), 3273-3291.

\bibitem {Re73}Renaud, A.~: Quelques propiet\'{e}s des applications
analytiques d'une boule de dimension infinie dans une autre. \emph{Bull. Sci. Math.}
\textbf{23}, 2 (1973) 129--159.

\bibitem {Sal05}Salgado, G. : Triple products on $\mathfrak{gl}_{n}$.
Submitted to publication.


\bibitem {Svi92}Svinolupov, S. I.: Generalized Schr\"{o}dinger equations
and Jordan pairs. \emph{Comm. Math. Phys.} \textbf{143} (1992), 559-575.



\bibitem {Up85}Upmeier, H.: \emph{Symmetric Banach manifolds and Jordan
}$C^{\ast}-$\emph{algebras.} North Holland Mathematical Studies \textbf{104}.
Elsevier, 1985.

\bibitem {Up92}Upmeier, H.: Jordan algebras, complex analysis and
quantization. In \emph{Jordan algebras (Oberwolfach, 1992)} 301-317. W. de
Gruyter, Berlin, 1994.
\end{thebibliography}
\end{document}